\documentclass{article}
\usepackage[margin=1in]{geometry}

\usepackage{multirow}
\usepackage{array}
\usepackage{longtable}

\usepackage[utf8]{inputenc}
\usepackage[T1]{fontenc}
\usepackage[utf8]{inputenc}
\usepackage{verbatim}
\usepackage{tikz}
\usetikzlibrary{shapes.geometric}
\usetikzlibrary{quotes}
\usetikzlibrary{arrows.meta}
\usepackage{subfig}

\usetikzlibrary{snakes,arrows,shapes}

\usepackage{hyperref}
\usepackage{bbm}
\usepackage{amsmath,amssymb,enumerate,ifthen,tikz,floatrow,amsmath, algorithm,algorithmic}
\usepackage{mathtools}
\usepackage{graphicx}

\graphicspath{ {./images/} }
\newtheorem{theorem}{Theorem} 
\newtheorem{lemma}{Lemma}

\newtheorem{corollary}{Corollary}

\newtheorem{remarka}{Remark} 
\usepackage{float}

\usepackage{hyperref}
 
\newenvironment{proof}{{\bf Proof.}}{\hfill\rule{2mm}{2mm}} 
\newenvironment{pproof}[1]{\noindent{\textbf{Proof of #1.}}}{\hfill\rule{2mm}{2mm}}

\newcommand{\boundellipse}[3]
{(#1) ellipse (#2 and #3)}

\title{Hierarchical Clustering: New Bounds and Objective}

\author{Mirmahdi Rahgoshay \\
Department of Computing Science\\ 
University of Alberta
\and Mohammad R. Salavatipour\footnote{Supported by NSERC.}\\
Department of Computing Science\\ 
University of Alberta}
\date{}

\begin{document}
\maketitle
\pagenumbering{arabic}
\begin{abstract}
Hierarchical Clustering has been studied and used extensively as a method for analysis of data. More recently, Dasgupta [STOC2016] initiated theoretical study and analysis of these clustering tools by defining
precise objective functions. In a hierarchical clustering, one is given a set of $n$ data points along with a notion of similarity/dis-similarity between them. More precisely for each two items $i$ and $j$ we are given a weight $w_{i,j}$ denoting their similarity/dis-similarity. The goal is to build a recursive (tree like) partitioning of the data points (items) into successively smaller clusters which is represented by a rooted tree where the leaves correspond to the items and each internal node corresponds to a cluster of all the items in the leaves of its subtree. Typically, the goal is to have the 
items that are relatively similar, to separate at deeper levels of the tree (and hence stay in the same cluster as deep as possible).
Dasgupta [STOC2016] defined a cost function for a tree $T$ to be $Cost(T) = \sum_{i,j \in [n]} \big(w_{i,j} \times |T_{i,j}| \big)$ where
$T_{i,j}$ is the subtree rooted at the least common ancestor of $i$ and $j$ and presented the first approximation algorithm for such clustering.
Then Moseley and Wang [NIPS2017] considered the dual of Dasgupta's objective function for similarity-based weights,
where the objective is to maximize $Rev_{Dual}(T) = \sum_{i,j \in [n]} \big(w_{i,j} \times (n - |T_{i,j}|) \big)$.
They showed that both random partitioning and average linkage have approximation ratio $1/3$ which has been improved to $0.336379$ [Charikar et al. SODA2019], $0.4246$ [Ahmadian et al. AISTATS2020], and more recently to $0.585$ [Alon et al. COLT2020].

Later Cohen-Addad et al. [JACM2019] considered the same objective function as Dasgupta's 
but for dissimilarity-based metrics: $Rev(T)=\sum_{i,j\in [n]} \big(w_{i,j}\times |T_{i,j}|\big)$,
where $w_{i,j}$ is the weight of dissimilarity between two nodes $i$ and $j$. In this version a good clustering should have larger $T_{i,j}$ when $w_{i,j}$ is relatively large. It has been shown that both random partitioning and average linkage have ratio $2/3$ which has been only slightly improved to $0.667078$ [Charikar et al. SODA2020]. 

Our first main result is to improve this ratio of 0.667078 for $Rev(T)$. We achieve this by building upon the earlier work and use a more delicate algorithm and careful analysis which can be refined to achieve approximation $0.71604$.

We also introduce a new objective function for dissimilarity-based Hierarchical Clustering. Consider any tree $T$, we define $H_{i,j}$ as the number of $i$ and $j$'s common ancestors in $T$. In other words, when we think of the process of building tree as a top-down procedure, $H_{i,j}$ is the step in which $i$ and $j$ are separated into two clusters (they were stayed within the same cluster for $H_{i,j}-1$ many steps). 
Intuitively, items that are similar are expected to remain within the same cluster as deep as possible and items that are dissimilar are to be separated into two different clusters higher up in the tree.
So, for dissimilarity-based metrics, it is better to separate two dissimilar items $i$ and $j$ at top levels and have low values of $H_{i,j}$, so we suggest the cost of each tree $T$, which we want to minimize, to be $Cost_H(T) = \sum_{i,j \in [n]} \big(w_{i,j} \times H_{i,j} \big)$. We present a $1.3977$-approximation for this objective.
\end{abstract}


\section{Introduction}
Hierarchical Clustering has been studied and used extensively as a method for analysis of data. Suppose we are given a set of $n$ data points (items)
along with a notion of similarity between them. 
The goal is to build a hierarchy of clusters, where each level of hierarchy is a clustering of the data points
that is a refined clustering of the previous level, and data points that are more similar stay together in deeper levels
of hierarchy. In other words, we want to output a 
 recursive partitioning of the items into successively smaller clusters, which are represented by a rooted tree, where the root corresponds to the set of all items, the leaves correspond to the items, and each internal node corresponds to the cluster of all the items in the leaves of its subtree. 
Many well-established methods have been used for Hierarchical Clustering, including some bottom-up agglomerative methods like single linkage, average linkage, and complete linkage, and some top-down approaches like the minimum bisection algorithm. In the bottom-up approaches, one starts from singleton clusters and at each step two clusters
that are more similar are merged. For instance, in the average linkage, the average of pair-wise similarity of points in two clusters is computed and clusters which have the highest 
average are merged into one and this continues until one cluster (of all points) is created. In the top-down approaches, one starts with a single cluster (of all points) and each
step a cluster is broken into two (or more) smaller ones. One such example is bisecting $k$-means (see \cite{Jia10}).
Although these methods have been around for a long time, it was only recently that researchers tried to formalize the goal and objective of hierarchical clustering.


Suppose the set of data points of input are represented as the vertices of a weighted graph $G=(V,E)$ where for any two nodes $i$ and $j$, $w_{i,j}$ is the weight (similarity or dissimilarity)
between the two data points. Then one can think of a hierarchical clustering as a tree $T$ whose leaves are nodes of $G$ and each internal node corresponds to the subset of nodes
of the leaves in that subtree (hence root of $T$ corresponds to $V$).
For any two data points $i$ and $j$ we use $T_{i,j}$ to denote the
subtree rooted at the least common ancestor (LCA) of $i$ and $j$ and $w_{i,j}$ represents
the similarity between $i,j$.
In the very first attempt to define a reasonable objective function, Dasgupta \cite{D16} suggested the cost of each tree $T$, which we want to minimize, to be:
\begin{equation}\label{min-sim}
Cost(T) = \sum_{i,j \in [n]} \big(w_{i,j} \times |T_{i,j}| \big).
\end{equation}

In other words, for each two items, $i$ and $j$, the cost to separate them at a step is the product of their weight and the size of the cluster at the time we separate them. Intuitively, this means that the clusters deeper in the tree would contain items that are relatively more similar. Dasgupta \cite{D16} proved that the optimal tree must be a binary tree. He then analyzed this objective function on some canonical examples (such as complete graph) and proved that it is $NP$-hard to find the tree with the minimum cost. Finally, he showed that a simple top-down heuristic graph partitioning algorithm, namely using taking the (approximately) minimum sparsest cut, would have a provably good approximation ratio.

An alternative interpretation of this cost function is in terms of cuts. In a top-down approach at each step we must partition a set of items into two groups (recall that the optimal tree is binary). We can set a cost for each step such that the total cost would be the summation of the costs of all the steps. If in one step we partition set $A \cup B$ of items into two sets $A$ and $B$, then the cost for this step would be $Cost(A,B) = \lvert A \cup B \rvert \times w(A, B)$ where $w(A,B)$ is the summation of all pairwise similarities between members of $A$ and $B$. Considering this, taking the minimum cut as the partition at each step seems a reasonable choice, although we will see later that this would not give a good approximation ratio. A nice property of this objective function is its modularity. More precisely, suppose $u$ is an internal node in tree $T$. If we replace $T_u$, the subtree rooted at $u$, by another subtree $T'_u$ containing the same set of items as leaves, and denote the new tree by $T'$, then the change in the total cost of the tree is only the difference between the costs of $T_u$ and $T'_u$: $Cost(T') = Cost(T) + Cost (T'_u) - Cost (T_u)$.



Dasgupta's showed that the top-down heuristic, which takes the minimum sparsest cut (approximately) at each step, has approximation factor of $O(\alpha \log n)$ \cite{D16}, where $\alpha$ is the best approximation ratio for minimum sparsest cut problem which is $O(\sqrt{\log n})$ \cite{ARV09}. 
Later Roy and Pokutta \cite{RP16} improved the previous result by giving an $LP$-based $O(\log n)$-approximation algorithm for the same objective function.
In more recent work Charikar and Chatziafratis \cite{CC17} showed that the algorithm of \cite{D16} in fact has approximation ratio of $O(\alpha)$.

Cohen-Addad et al.~\cite{CKMM19} considered the same objective function but for dissimilarity-based graphs, 
where $w_{i,j}$ is the weight of dissimilarity between two nodes $i$ and $j$. In this version a good clustering should have larger $T_{i,j}$
when $w_{i,j}$ is relatively large. Here the objective is to maximize the following formula:
\begin{equation}\label{max-dissim}
	Rev(T) = \sum_{i,j \in [n]} \big(w_{i,j} \times |T_{i,j}| \big)
\end{equation}

They showed that the random top-down partitioning algorithm is a $2/3$-approximation 
and the classic average-linkage algorithm gives a factor $1/2$ approximation 
(later \cite{CCN19} mentioned that the same analysis will show that it is actually $\frac{2}{3}$-approximation), 
and provided a simple top-down local search algorithm that gives a factor 
$(\frac{2}{3} - \epsilon)$-approximation.
They take an axiomatic approach for defining `good' objective functions 
for both similarity and dissimilarity-based hierarchical clustering by characterizing a set of admissible objective functions (that includes the one introduced by Dasgupta) that have the property that 
when the input admits a `natural' ground-truth hierarchical clustering, the ground-truth clustering has an optimal value.
They also provided a similar analysis showing that 
the algorithm of Dasgupta \cite{D16} (using sparsest cut algorithm) has ratio of $O(\alpha)$; their analysis
of this is different (and slightly better) than \cite{CC17}.
More recently Chatziafratis et al.~\cite{CGL20} showed that it is hard to approximate $Rev(T)$
within a constant of $9159/9189 = 0.996735$ assuming the Unique Games Conjecture.

\subsection{More Related Work}

Moseley and Wang \cite{MW17} considered the dual of Dasgupta's objective function for similarity-based graphs,
where the objective is to maximize the following formula:
\begin{equation}\label{max-sim}
	Rev_{Dual}(T) = \sum_{i,j \in [n]} \big(w_{i,j} \times (n - |T_{i,j}|) \big)
\end{equation}

For each tree $T$ we have $Cost(T) + Rev_{Dual}(T) = nW$ where $W$ is the summation of the weights of similarity 
over all pairs of the items which is not dependent to the structure of the tree.
This means that the optimum solution for both objective functions is the same.
They showed that the classic average-linkage algorithm 
as well as the random top-down partitioning algorithm have approximation ratio $1/3$ 
and provided a simple top-down local search algorithm 
that gives a $(\frac{1}{3}- \epsilon)$-approximation.

Later Charikar, Chatziafratis and Niazadeh \cite{CCN19} proved that the average-linkage algorithm 
is tight for both objective functions (\ref{max-dissim}) and (\ref{max-sim}).
They also gave two top-down algorithms to beat the average-linkage ratios.
More specifically for maximizing $Rev_{Dual}(T)$ (\ref{max-sim}),
they provided an $SDP$-based algorithm which has approximation ratio $0.336379$ (slightly better than $1/3$).
For the maximizing dissimilarity-based graphs (\ref{max-dissim}), they gave a top-down algorithm 
with a factor $0.667078$ approximation (slightly better than $2/3$).
We will go through this algorithm in more details shortly as one of our results is to improve their approach.

More recently, Ahmadian et al.~\cite{ACELMMY20} provided a $0.4246$-approximation algorithm 
for maximizing $Rev_{Dual}(T)$.
What they do is to detect the cases where average linkage is not good and show that in those cases
the maximum uncut bisection would gain a good fraction of the objective of the optimum solution in the very first step.
So, by taking the better of the two of average-linkage and maximum uncut bisection 
(if we can solve it optimally in polynomial time) in the first step and average linkage for the remaining steps
the approximation ratio would be $4/9$.
But the best known algorithm for maximum uncut bisection has approximation ratio $\rho = 0.8776$ \cite{ABG16}, 
so the ratio of their algorithm decreases from $4/9$ to $0.4246$ 
which is still much better than the previous best $0.3363$-approximation of \cite{CCN19}. 
They also complemented their positive results by providing the APX-hardness
(even for 0-1 similarities), under the Small Set Expansion hypothesis \cite{RS10}.

More recently, Alon et al.~\cite{AAV20} proved that the algorithm of \cite{ACELMMY20} is actually giving a $2\rho/3=0.585$-approximation by proving the existence of a better maximum uncut bisection.
This is considered the third improvement over the Average-Linkage for revenue maximization of similarity-based Hierarchical Clustering (\ref{max-sim}), while the best algorithm for the revenue maximization of dissimilarity-based version (\ref{max-dissim}) is still only slightly better than the average linkage ($0.667$ vs $2/3$).

Chatziafratis et al.~\cite{CNC18} considered a version of the problem where we have 
some prior information about the data that imposes constraints on the clustering hierarchy
and provided provable approximation guarantees for two simple top-down algorithms on similarity-based graphs.
More recently, Bakkelund \cite{B20} considered 
order preserving hierarchical agglomerative clustering which is
a method for hierarchical clustering of directed acyclic graphs and 
other strictly partially ordered data that preserves the data structure.

Emamjomeh-Zadeh and Kempe \cite{ED18} considered adaptive Hierarchical Clustering using the notion of ordinal queries,
where each ordinal query consists of a set of three elements, and the response to a query reveals 
the two elements (among the three elements in the query) which are “closer” to each other.
They studied active learning of a hierarchical clustering using only ordinal queries 
and focused on minimizing the number of queries even in the presence of noise.

Wang and Wang \cite{WW20} suggested that Dasgupta's cost function is only effective 
in differentiating a good HC-tree from a bad one for a fixed graph, 
But the value of the cost function does not reflect how well an input similarity graph 
is consistent with a hierarchical structure and present a new cost function, 
which is based on Dasgupta's cost function but gives a cost between $0$ and $1$ to each tree.

Charikar et al.~\cite{CCNY19} were the first one to consider Hierarchical Clustering for Euclidean data and
showed an improvement is possible for similarity-based graphs on objective (\ref{max-sim}). 
Later Wang and Moseley \cite{WM20} considered objective (\ref{max-dissim}) for Euclidean data
and showed that every tree is a $1/2$-approximation if the distances form a metric and 
developed a new global objective for hierarchical clustering in Euclidean space and
proved that the optimal 2-means solution results in a constant approximation for their objective.

Hogemo et al.~\cite{HPT20} considered the Hierarchical Clustering of unweighted graphs
and introduced a proof technique, called the normalization procedure, 
that takes any such clustering of a graph $G$ and iteratively improves it 
until a desired target clustering of $G$ is reached.
More recently Vainstein et al.~\cite{VCCRMA21} proved structural lemmas for both objectives (\ref{max-sim}) and (\ref{max-dissim})
allowing to convert any HC tree to a tree with constant number of internal nodes while incurring an arbitrarily small loss.
They managed to obtain approximations arbitrarily close to $1$, 
if not all weights are small (i.e., there exist constants $\epsilon$ and $\delta$ 
such that the fraction of weights smaller than $\delta$, is at most $1-\epsilon$).

Chehreghani \cite{H20} proposed a hierarchical correlation clustering method 
that extends the well-known correlation clustering to produce hierarchical clusters.
Later Vainstein et al.~\cite{VCCRMA21} provided a $0.4767$-approximation and 
presented nearly optimal approximations for complementary similarity/dissimilarity weights.
There are also some other works including 
the ones trying to reduce the time of the current algorithms \cite{MVW21} and \cite{DELMS21}, 
and those introducing Fair Hierarchical Clustering \cite{AEK0MMPV20} and Online Hierarchical Clustering \cite{KRSCCK19}.

\subsection{Our Results}
Our first result is to consider the revenue maximization of dissimilarity (i.e. objective (\ref{max-dissim})) and improve upon the algorithm of \cite{CCN19} which has ratio 0.667078:

\begin{theorem}\label{hc:max-dissim}
For hierarchical clustering on dissimilarity-based graphs, there is an approximation algorithm to maximize objective of (\ref{max-dissim}) with ratio $0.71604$.
\end{theorem}

To prove this, we build upon the work of \cite{CCN19} and present an algorithm that takes advantage of some conditions
in which their algorithm fails to perform better. Since the final algorithm (and its analysis) is more complex, and
for ease of exposition, we start with a simpler algorithm which achieves ratio $0.6929$. Then through a series of
improvements we show how we can get to $0.71604$ ratio.

Next we introduce a new objective function for hierarchical clustering and present approximation algorithms for this new objective.
The intuition for this new objective function is that we expect the items that are more similar to remain
in the same cluster for more steps, i.e. the step in which they are separated into different clusters is
one which is further away from the root of the tree.
Consider any tree $T$, we define $H_{i,j}$ as the number of common ancestors of $i$ and $j$ in $T$.
In other words, when we think of the process of building the tree as a top-down procedure, 
$H_{i,j}$ is the step in which we decide to separate $i$ and $j$ from each other (they were decided to be together in $H_{i,j}-1$ many steps).
Intuitively, similar items are supposed to stay together until deeper nodes in the tree and hence are expected to have larger $H_{i,j}$ values
whereas dissimilar items are supposed to be separated higher up in the tree. For instance, looking at a phylogenetic or genomic tree, the species that are most dissimilar
are separated higher up in the tree (i.e. have small $H_{i,j}$) whereas similar species are separated at deep nodes of the tree and hence have high $H_{i,j}$ values.
For dissimilarity-based graphs, we propose to minimize the following objective:
\begin{equation}\label{dissim-hij}
	Cost_H(T) = \sum_{i,j \in [n]} \big(w_{i,j} \times H_{i,j} \big).
\end{equation}

The problem we are looking to solve here is to find a full binary tree with the minimum $Cost_H(.)$.
It is easy to see that any algorithm 
that gives a balanced binary tree would have approximation ratio at most $O(\log n)$ since the height of such trees is $O(\log n)$. Furthermore, it is not hard to verify that the average-linkage algorithm would not perform well for this new objective function.
The following example is an instance for which the cost of the solution of the average-linkage algorithm is at least 
$O(\frac{n}{\log n})$ times the cost of the optimum solution.
Consider a graph with $n$ vertices: $v_1, v_2, ..., v_n$. 
Then for each $2 \leq j \leq n$ and for each $1 \leq i < j$ let $w_{i,j} = j-1$.
In this graph the summation of all the edges would be $O(n^3)$ but running average linkage on this graph
would result in a tree $T$ with $Cost_H(T) = O(n^4)$ while the optimum tree (as well as any balanced binary tree) will have cost $O(n^3 \log n)$.

Our second main result is the following:

\begin{theorem}\label{hc:thm:h}
	For hierarchical Clustering on dissimilarity-based graphs, a top-down algorithm that
	chooses the approximated weighted max-cut at each step, 
	would be a $\frac{4\alpha_{GW}}{4\alpha_{GW}-1}$-approximation algorithm to minimize $Cost_H(T)$, 
	where $\alpha_{GW}$ is the ratio of the max-cut approximation algorithm.
\end{theorem}

Considering that the best known approximation algorithm for weighted maximum cut problem has ratio $\alpha_{GW} = 0.8786$ \cite{GW95},
the ratio of this algorithm would be $1.3977$. This also means that any top-down algorithm which cuts at least half of the weight of the remaining edges,
including the random partitioning algorithm, would have approximation ratio $2$.

\section{Maximizing $Rev(T)$ in Dissimilarity-Based Graphs}
Our goal in this section is to prove Theorem \ref{hc:max-dissim}. For ease of exposition, we first present an algorithm
and prove it has approximation $0.6929$. Then we show how running the same algorithm with different
parameters and taking the best solution of all we obtain a 0.71604-approximation.
Our algorithm for maximizing $Rev(T)$ (objective function (\ref{max-dissim})) builds upon the algorithm of Charikar et al. \cite{CCN19} which has approximation ratio
$0.667078$, slightly better than random partition at each step, which has ratio $2/3$.

As explained in \cite{CCN19}, one can see that the top-level cuts of the tree, i.e. those corresponding to clusters closer to the root are making the large portion of the optimum value.
Therefore, it seems reasonable, in a top-down approach, to use larger cut sizes at each step. So, this suggests a simple algorithm: at each step try to find a max-cut (or approximate max-cut). However, this
 “recursive max-cut” fails. As shown in \cite{CCN19}, for a graph of $n$ vertices with a clique of size $\epsilon n$ where the rest of the edges have weight zero, the optimum solution "peels off" vertices of the implanted clique one by one (i.e. in initial steps each vertex of the clique is separated from the rest), and
this obtains an objective value of at least $n(1 - \epsilon)W$, where $W$ is the sum of all edge weights. But the recursive max-cut (even if we find the optimum max-cut in each step) will have ratio 
at most $\frac{2+\epsilon}{3} nW$. 
Thus, this "recursive max-cut" is not going to perform better than the trivial random partitioning at each step.

Inspired by this, \cite{CCN19} suggested an algorithm that initially will peel off vertices with high (weighted) degree one by one
(depending on a predetermined threshold) and after that will use a max-cut in one step to partition the remaining cluster into two.
From there on, we can assume we use the random partitioning. This is the "peel-off first, max-cut next" algorithm of \cite{CCN19}.
Note that the upper bound used so far in previous works for optimum is $nW$. Intuitively, if the optimum value is close to this lower bound
then there must be a good max-cut, and if the optimum is bounded away from this then random partition performs better than 2/3.
They showed that the better of this "peel-off first, max-cut next" algorithm and the random partitioning algorithm has approximation ratio of $0.667078$ 
for maximizing $Rev(T)$ in dissimilarity-based graphs, which slightly beats the $2/3$-approximation of
simply doing random partitioning.

Our algorithm is based on similar ideas \cite{CCN19} but has more steps added into it to improve the bound.
Our basic algorithm takes the better of the two of 
"Random Partitioning" algorithm and our "Peel-Off first, Max-Cut or Random Next" algorithm.
This is to make sure in the cases where the revenue of the optimum solution is too far from $nW$,
the random partitioning algorithm would give us a good enough approximation 
and we can focus on the case where the revenue of the optimum, 
which we denote by $OPT$, is close to $nW$.

A key difference between our "Peel-Off first, Max-Cut or Random Next" algorithm and algorithm
of \cite{CCN19}, which is "Peel-Off first Max-Cut Next" algorithm, 
is that after the Peel Off in the first phase we then choose the better of the two of 
"Random Partitioning" and "Max Cut" algorithms. 
Actually, by looking at the fraction of $W$ which is peeled off in the first phase, 
we can decide which one of the "Random Partitioning" or "Max Cut" would be good enough 
for the second phase and only go ahead with that one.
We have also generalized their analysis by considering a few more parameters to 
finally improve the approximation ratio from $0.667078$ to $0.6929$.

In our algorithm, we first set a parameter $\gamma \geq 1$ and start the Peel-Off process.
More precisely we define $W_v$ for each remaining vertex as 
the current total of the weights of the edges incident to $v$
and remove vertices with $W_v \geq \gamma \frac{2W}{n}$ and all their edges, 
where $W$ is the total of the weight of all the edges (not only the remaining edges) 
and $n$ is the total number of vertices. So, $W_v$ is a dynamic value (this is different from what \cite{CCN19} do as they compare
the initial $W_v$ with the threshold value).
After removing one vertex (which is always the one with the largest $W_v$) 
we remove all the edges incident to it and update $W_v$ for the remaining vertices accordingly.
Note that, we do not update $W$ and $n$ and these two are fixed (to the initial values) throughout the entire peel-off process.
So our Peel-Off process is a bit different than the one used in \cite{CCN19}
as we need to update $W_v$ for the remaining vertices after removing each vertex.

After we reach a state where all the remaining vertices have $W_v < \gamma \frac{2W}{n}$,
we then look at the fraction of the $W$ which is peeled off.
Let us denote all the peeled off vertices by $V_R$ and call them "Red" vertices (this is similar to the terminology used by \cite{CCN19}).
We also denote all the edges incident to at least one red vertex,
those that are removed from the graph after the peel off phase, by red edges, $E_R$, and denote the total weight of all the red edges by $W_R$ and define $R = W_R/W$ as the weighted fraction of the red edges.
We also call the vertices and edges that are not red, as blue vertices and edges, respectively:
$V_B = V \setminus V_R$ and $E_B = E \setminus E_B$.

At this stage of our algorithm if $R$ is greater than a pre-defined parameter $0 < R^* < 1/2$ (to be chosen later)
this means that a good fraction of the vertices are peeled-off 
and we continue the algorithm by doing the random partitioning.
Otherwise, this means that remaining graph is pretty dense, and we can prove there should be a big cut.
So, if $R \leq R^*$ we continue the algorithm by doing an approximated max-cut. After max-cut we do the
rest using random partitioning.
Formally we propose "Peel-off First, Max Cut or Random Next" (Algorithm \ref{hc:alg:peel-off-max-cut-random})
and prove Theorem \ref{hc:thm:DissimMax}.


\begin{algorithm}
	\begin{algorithmic}
		\STATE \textbf{Input:} $G = (V, E)$, dissimilarity weights$\{w_{i,j}\}_{(i,j)\in E}$, and parameters $\gamma \geq 1$ and $0 < R^* < 1/2$
		\STATE define $W = \sum_{(v,u) \in E} w_{v,u}$ and $n = \lvert V \rvert$.
		\STATE Initialize hierarchical clustering tree $T \leftarrow \emptyset$.
		\STATE Initialize $V_B \leftarrow V$ and $E_B \leftarrow E$.
		\WHILE{there exists a vertex $v \in V_B$ with $W_v = \sum_{u \in V_B: (v,u) \in E_B} w_{v,u} \geq \gamma \frac{2W}{n}$}
		\STATE Choose $v^* \in V_B$ with the largest $W_v$.
		\STATE Update $T$ by adding the cut $(\{v^*\}, V_B \setminus \{v^*\})$.
		\STATE Update $V_B$ by removing $v^*$; $V_B \leftarrow V_B \setminus \{v^*\}$.
		\STATE Update $E_B$ by removing all the edges incident to $v^*$; $E_B \leftarrow E_B \setminus \{e \in E_B : e$ incident to $v^*\}$.
		\ENDWHILE
		\STATE define $W_R = W - \sum_{(v,u) \in E_B} w_{v,u}$ and $R = \frac{W_R}{W}$.
		\IF{$R > R^*$}
		\STATE Recursively run Random Partitioning Algorithm on $V_B$ and update $T$.
		\ELSE
		\STATE Run approximate Max-Cut \cite{GW95} on $G_B = (V_B, E_B)$
		\STATE Let the resulting cut be $(V_L, V_R)$ and update $T$ by adding this cut.
		\STATE Recursively run Random Partitioning Algorithm on $V_L$ and $V_R$ and update $T$.
		\ENDIF
		\RETURN $T$
	\end{algorithmic}
	\caption{Peel-off First, Max Cut or Random Next}
	\label{hc:alg:peel-off-max-cut-random}
\end{algorithm}

\begin{theorem}\label{hc:thm:DissimMax}
	For Hierarchical Clustering on dissimilarity-based graphs, 
	there exists a choice of $\gamma \geq 1$ and $0 < R^* < 1/2$ such that 
	the better of Algorithm \ref{hc:alg:peel-off-max-cut-random} and "Random Partitioning" algorithm
	would be an $\alpha$-approximation algorithm to maximize $Rev(T)$, where $\alpha = 0.6929$.
\end{theorem}

\subsection{Proof of Theorem \ref{hc:thm:DissimMax}}
To prove Theorem \ref{hc:thm:DissimMax}, we first need to have some more definitions.
First recall that in an instance with $n$ vertices and the total weights of dissimilarities of $W$, 
the random partitioning algorithm always has revenue at least $2nW/3$ \cite{CNC18}.

In the cases where $OPT$ is too far from $nW$ the random partitioning algorithm would be good enough.
More precisely we consider $OPT = (1-\epsilon) nW$, and this defines $\epsilon$, which is not known to us.
If $\epsilon \geq \frac{\alpha-2/3}{\alpha}$, then $OPT \leq \frac{2}{3\alpha} nW$ 
and the random partitioning algorithm would be an $\alpha$-approximation. So, we assume $\epsilon < \frac{\alpha-2/3}{\alpha}$ and 
prove that there exist choices of $\gamma$ and $R^*$ which make Algorithm \ref{hc:alg:peel-off-max-cut-random}
an $\alpha$-approximation for $\alpha = 0.6929$. 

\begin{lemma} \label{hc:lm:ell}
Let the number of vertices that are peeled off after the first phase of our algorithm be $\lvert V_R \rvert = n\ell$. Then
	$\ell \leq \frac{R}{2\gamma}$
\end{lemma}
\begin{proof}
Recall that only those vertices 
that have $W_v \geq \gamma \frac{2W}{n}$ at the time of the peel-off are peeled off in the first phase.
So, we can say that $W_R$, which is the summation of $W_v$ of all the peeled off vertices,
is at least $n\ell \cdot \gamma \frac{2W}{n}$.
This means $W_R \geq 2\gamma \ell W$.
\end{proof}

Now assume $ALG_P$ is the contribution of all the red edges (of our algorithm) to the final objective revenue $Rev(T)$.
This is the revenue which is obtained in the peel-off phase. The following lemma is a stronger version of Lemma 5.1 of \cite{CCN19}.

\begin{lemma} \label{hc:lm:algP}
	$ALG_P \geq (1-\ell/2) nW_R$
\end{lemma}
\begin{proof}
Suppose $V_R = \{v_1, v_2, ..., v_{n\ell}\}$ where $v_i$ is the $i$'th vertex 
which is peeled off.
Recall that at each step of the peel-off phase of the algorithm we choose the vertex with the largest $W_v$.
Thus $W_{v_1} \geq W_{v_2} \geq ... \geq W_{v_{n\ell}} \geq \gamma \frac{2W}{n}$.
Now consider the revenue which is obtained in the peel-off phase. 
While removing $v_i$ the size of the tree is $n-i+1$, so we have:
\vspace{-3mm}
\begin{eqnarray*}
	ALG_P &=& n W_{v_1} + (n-1) W_{v_2} + ... + (n - n\ell + 1) W_{v_{n\ell}}\\
	&\geq& \frac{\big(\sum_{i = n-n\ell+1}^n i \big) \cdot \big(\sum_{j=1}^{n\ell} W_{v_j}\big)}{n\ell}\\
	&=& \big(\frac{n\ell \cdot (n + (n-n\ell+1))/2}{n\ell} \big) \cdot W_R\\ 
	&=& (1-\ell/2 + 1/n) \cdot nW_R \\
	&\geq& (1-\ell/2) nW_R.
\end{eqnarray*}

\end{proof}

As we mentioned earlier, after the peel-off phase, the algorithm will look at the ratio $R=W_R/W$ and compares it
with the given parameter $R^*$. If $R>R^*$, the second phase will be to do the random partitioning and no max-cut is needed.
Otherwise, if $R \leq R^*$, then the algorithm will use Goemans and Williamson's algorithm 
for Max-Cut \cite{GW95} on the remaining graph (blue vertices and edges)
to do one partition (the rest can be done in any arbitrary manner, 
say random partition).
So, we separate our analysis in those two cases and prove necessary lemmas and theorems in each case.

\subsubsection{Case 1: $R > R^*$}
Let us denote the total revenue obtained in the random partitioning phase of the algorithm by $ALG_R$.
Then the total revenue of our algorithm would be at least $ALG_P + ALG_R$.
For this case of $R>R^*$ we prove the following theorem:

\begin{theorem}\label{hc:thm:algPR}
	If $R > R^*$:
		$ALG_P + ALG_R \geq nW \cdot (\frac{2}{3} + \frac{\gamma-1}{3\gamma} \cdot R^* + \frac{1}{12\gamma} \cdot {R^*}^2).$
\end{theorem}

This means that the ratio of our algorithm in this case is at least
$(\frac{2}{3} + \frac{\gamma-1}{3\gamma} \cdot R^* + \frac{1}{12\gamma} \cdot {R^*}^2)$.
Note that to have this ratio, we are actually comparing our revenue with $nW$
(and not $OPT$ which might in fact be smaller)
and if we somehow find a better upper-bound for $OPT$, 
comparing our revenue with $OPT$ might prove a better ratio for our algorithm.
To prove this Theorem, we first recall that in an instance with 
$n$ vertices and the total weights of dissimilarities of $W$, 
random partitioning algorithm has revenue at least $2nW/3$ \cite{CNC18}.

\begin{lemma} \label{hc:lm:algR}
	$ALG_R \geq \frac{2}{3}(1-\ell) n (W - W_R)$.
\end{lemma}
\begin{proof}
This follows easily from the observation mentioned above since after the peel-off phase the number of remaining vertices is
$\lvert V_B \rvert = n - n\ell=(1-\ell)n$ and 
the total weights of the remaining edges is $W_B = W - W_R$.
\end{proof}

The rest of the proof of Theorem \ref{hc:thm:algPR} will be done 
by simply using Lemmas \ref{hc:lm:ell}, \ref{hc:lm:algP}, \ref{hc:lm:algR}, and the fact that $R > R^*$.
Combining Lemmas \ref{hc:lm:algP} and \ref{hc:lm:algR} we have:
\vspace{-3mm}
\begin{eqnarray*}
	ALG_P + ALG_R &\geq& (1-\ell/2) nW_R + \frac{2}{3}(1-\ell) n (W - W_R) \\
	\Rightarrow \frac{ALG_P + ALG_R}{nW} &\geq& (1-\ell/2) R + \frac{2}{3}(1-\ell) (1-R).
\end{eqnarray*}

The RHS is decreasing with $\ell$ increasing, so using Lemma \ref{hc:lm:ell}, we have:

\begin{equation*}
	\Rightarrow \frac{ALG_P + ALG_R}{nW} \geq (1-\frac{R}{4\gamma}) R + \frac{2}{3}(1-\frac{R}{2\gamma}) (1-R) = (\frac{2}{3} + \frac{\gamma-1}{3\gamma} \cdot R + \frac{1}{12\gamma} \cdot R^2).
\end{equation*}

This bound is increasing with $R$, so considering the fact that $R > R^*$, it implies Theorem \ref{hc:thm:algPR}.

Let us define function $F(R) = \left(1-\frac{R}{4\gamma}\right) R + \frac{2}{3}\left(1-\frac{R}{2\gamma}\right) (1-R)$ 
and note that in Theorem \ref{hc:thm:algPR} we just proved that if $R > R^*$ then $ALG_P + ALG_R \geq nW\cdot F(R^*)$. We will use this bound later.

\subsubsection{Case 2: $R \leq R^*$}
Like the other case, let us denote the total revenue obtained in the max-cut phase of the algorithm by $ALG_C$.
Then the total revenue of our algorithm would be at least $ALG_P + ALG_C$. Note that we are not considering the revenue obtained after the
max-cut phase, which uses random partitioning (as we do not have a good bound on the value/weight of the edges left after performing the max-cut).
For this case of $R\leq R^*$ we prove the following theorem:

\begin{theorem}\label{hc:thm:algPC}
Suppose $R \leq R^*$ and the following conditions for $\epsilon$, $R^*$, and $\gamma$ hold (where $\alpha_{GW}$ is the ratio of
approximation for max-cut \cite{GW95}):
\vspace{-3mm}
\begin{eqnarray}
R^* &<& 1 - \frac{6\alpha_{GW}}{3\alpha_{GW}-2} \cdot \epsilon.	 \label{hc:cond:Re}\\
\gamma &\leq& (1-R^*)^2 \cdot \frac{3\alpha_{GW}-2}{9\alpha_{GW}^2} \cdot \frac{1}{\epsilon}. \label{hc:cond:Reg}\\
\frac{1}{3} &\geq& R^* \cdot \left(\frac{\alpha_{GW}}{2}-\frac{1}{6}\right) + \gamma \cdot (1-\alpha_{GW}) .\label{hc:cond:Rg}
\end{eqnarray} 
Then $ALG_P + ALG_C \geq nW \cdot F(R^*).$
\end{theorem}

This means that the ratio of our algorithm in this case (like the other case) is at least $F(R^*)$
if conditions (\ref{hc:cond:Re}), (\ref{hc:cond:Reg}), and (\ref{hc:cond:Rg}) are met.

We continue by looking at the layered structure of the optimum tree similar to the one done in \cite{CCN19}.
Fix parameter $\delta < 1/2$ (to be specified later) and imagine we start from the root of the optimum tree.
\noindent\parbox[b][][s]{0.5\linewidth}{%
We are looking for the smallest (i.e. deepest) cluster of the tree with size more than $n(1-\delta)$. This is a cluster
where itself and all its ancestors have size more than $n(1-\delta)$ but both its children are smaller.
Consider each time the optimum tree performs a partition of a cluster into two.
Since $\delta < 1/2$, either both of these clusters must be of size at most $n(1-\delta)$ (at which point we stop),
or exactly one of them is smaller than $n(1-\delta)$ and the other is strictly larger than $n(1-\delta)$.
In the latter case, we go down the bigger branch and keep doing this until both children have size at most $n(1-\delta)$. At this point we have found the smallest 
}
\hfill
    \includegraphics[scale=.55]{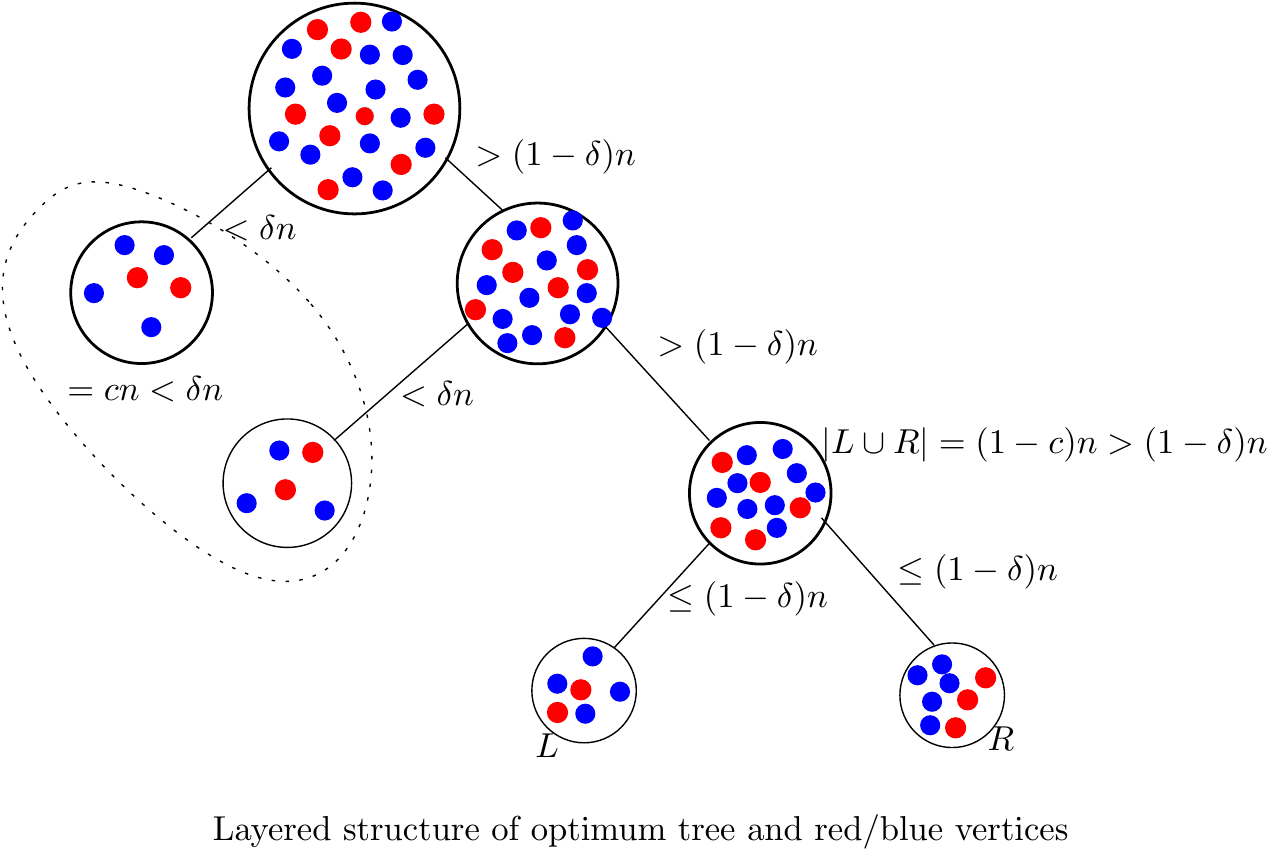}
    
\noindent    
cluster (internal node) of the tree with size more than $n(1-\delta)$,
which then has two (children) clusters of size at most $n(1-\delta)$.
We denote these two children clusters by $L$ and $R$ and we have:
$\lvert L \rvert < n(1-\delta)$, $\lvert R \rvert < n(1-\delta)$,
$\lvert L \cup R \rvert > n(1-\delta)$, and finally $\lvert V \setminus (L \cup R) \rvert < n\delta$.
We define $c = \lvert V \setminus (L \cup R) \rvert / n$ and we know $0 \leq c < \delta$.

Now we partition the blue vertices (vertices that survived the peel off phase) into two groups.
$V_{B-Cut} = V_B \cap (L \cup R)$ and $V_{B-Chain} = V_B \setminus (L \cup R)$.
Note that $V_{B-Cut}$ are those inside the smallest cluster (internal node) 
of the optimum tree with size more than $n(1-\delta)$ and $V_{B-Chain}$ are those outside it.
We also define $W_{B-Chain}$ as the total weights of the blue edges incident to at least one vertex of $V_{B-Chain}$.
Now $W_B$ is partitioned into four part: $W_B = W_L + W_R + W_{L,R} + W_{B-Chain}$, 
where $W_{L,R}$ are those with one end in $L$ and the other in $R$ 
while $W_L$ and $W_R$ are those with both ends in $L$ and $R$, respectively.
Now we are ready for our first lemma.

\begin{lemma} \label{hc:lm:WBCh}
	$W_{B-Chain} \leq 2c\gamma W$.
\end{lemma}
\begin{proof}
We know that $\lvert V_{B-Chain} \rvert \leq \lvert V \setminus (L \cup R) \rvert = cn$.
We also know that each $v \in V_{B-Chain}$ has $W_v < \gamma \frac{2W}{n}$:
\vspace{-3mm}
\begin{equation*}
	W_{B-Chain} \leq \sum_{v \in V_{B-Chain}} W_v < \lvert V_B \rvert \cdot \gamma \frac{2W}{n} = 2c\gamma W.
\end{equation*}\vspace{-1.5mm}
\end{proof}

Now it iss time to find a cut in $V_B$ with provably good size.
This layered structure of the optimum tree would suggest the cut $(L, R)$ 
with size $W_{L,R}$ and the following lemma will give us a lower-bound on its size:

\begin{lemma} \label{hc:lm:WLR}
	$W_{L,R} \geq \frac{\delta W_B - (\epsilon + 2c\delta\gamma) W}{\delta-c}$.
\end{lemma}
\begin{proof}
Consider the revenue of the optimum solution. By looking at the layered structure we just introduced we have:

\begin{equation*}
	OPT \leq n (W_R + W_{B-Chain}) + n(1-\delta) (W_L + W_R) + n(1-c) W_{L,R}.
\end{equation*}

This is because the size of the cluster when the optimum solution produces cut $(L, R)$ is exactly $n(1-c)$ 
and after that the size of the following clusters could not be more than $n(1-\delta)$.
Now we use $OPT = (1-\epsilon) nW$ and $W_L + W_R = W_B - W_{B-Chain} - W_{L,R}$ in the above inequality:
\vspace{-1.5mm}
\begin{equation*}
	OPT = (1-\epsilon) nW \leq n (W_R + W_{B-Chain}) + n(1-\delta) (W_B - W_{B-Chain} - W_{L,R}) + n(1-c) W_{L,R}.
\end{equation*}

By dividing both sides by $n$ and considering the fact that $W = W_R + W_B$ we have:

\begin{equation*}
	\Rightarrow W - \epsilon W \leq W - \delta (W_B - W_{B-Chain} - W_{L,R}) - c W_{L,R}.
\end{equation*}
\begin{equation*}
	\Rightarrow \delta(W_B - W_{B-Chain}) - \epsilon W \leq (\delta-c) W_{L,R}.
\end{equation*}

Considering the fact that $c$ is strictly less than $\delta$ and using Lemma \ref{hc:lm:WBCh} this completes the proof of the lemma.
\end{proof}

The next lemma will provide a lower bound on $ALG_C$.

\begin{lemma} \label{hc:lm:ALGc1}
	$ALG_C \geq \alpha_{GW} \cdot (1-\ell) n W_{L,R}$.
\end{lemma}
\begin{proof}
Note that after the peel-off phase, we are going to use 
Goemans and Williamson's max-cut algorithm to find a cut in $V_B$ whose size is $n-n\ell$.
As we know there is a cut with size at least $W_{L,R}$, we can find one with size $\alpha_{GW} W_{L,R}$ 
and because there are $(1-\ell)n$ many blue vertices left after the first phase,
the revenue we will have on the first step of the second phase will be at least $\alpha_{GW} \cdot (1-\ell) n W_{L,R}$.
\end{proof}

Note that if we can somehow find a lower-bound for the revenue which is obtained after this max-cut step,
(perhaps using a bi-section instead of a max-cut)
our ratio could be even better. We should point out that using a bi-section \cite{ACELMMY20} obtained
the improved ratio for the similarity-based version.

Considering the fact that $R = W_R/W$, $W = W_R + W_B$ and combining the previous two lemmas we have:

\begin{equation}
ALG_C \geq \alpha_{GW} \cdot (1-\ell) n W \left(\frac{\delta (1-R) - (\epsilon + 2c\delta\gamma))}{\delta-c}\right).\label{hc:lm:ALGc2}
\end{equation}

Now we want to set $\delta = \frac{3\alpha_{GW}}{3\alpha_{GW}-2} \cdot \frac{\epsilon}{1-R^*}$, and to be able to do that we need Condition (\ref{hc:cond:Re}).
That is because we have to make sure $\delta < 1/2$ and Condition (\ref{hc:cond:Re}) will be enough to have that:

\begin{equation*}
	R^* < 1 - \frac{6\alpha_{GW}}{3\alpha_{GW}-2} \cdot \epsilon \quad \Rightarrow \quad \frac{6\alpha_{GW}}{3\alpha_{GW}-2} \cdot \epsilon < 1 - R^* \quad \Rightarrow \quad \frac{3\alpha_{GW}}{3\alpha_{GW}-2} \cdot \frac{\epsilon}{1-R^*} < 1/2
\end{equation*}

Now we apply $\delta = \frac{3\alpha_{GW}}{3\alpha_{GW}-2} \cdot \frac{\epsilon}{1-R^*}$ to Equation (\ref{hc:lm:ALGc2}):

\begin{corollary} \label{hc:lm:ALGc3}
	$ALG_C \geq \alpha_{GW} \cdot (1-\ell) n W\cdot G(c)$ where
	$G(c) = \left(\frac{(1-R) - (1-R^*)\frac{3\alpha_{GW}-2}{3\alpha_{GW}} - 2c\gamma)}{1-c \frac{(1-R^*)}{\epsilon} \frac{3\alpha_{GW}-2}{3\alpha_{GW}}}\right).$
\end{corollary}

Note that $0 \leq c < \delta = \frac{3\alpha_{GW}}{3\alpha_{GW}-2} \cdot \frac{\epsilon}{1-R^*}$,
and the following lemma would make sure the worst case for us is when $c = 0$ (see Appendix \ref{app:proofs} for proof).

\begin{lemma} \label{hc:lm:c}
	If Condition (\ref{hc:cond:Reg}) holds, then $G(0) \leq G(c)$ for all $0 \leq c < \delta = \frac{3\alpha_{GW}}{3\alpha_{GW}-2} \cdot \frac{\epsilon}{1-R^*}$.
\end{lemma}

Using this lemma and applying $c=0$ to Corollary \ref{hc:lm:ALGc3} we have:

\begin{equation}
ALG_C \geq \alpha_{GW} \cdot (1-\ell) n W \left((1-R) - (1-R^*)\frac{3\alpha_{GW}-2}{3\alpha_{GW}}\right).\label{hc:lm:ALGc4}
\end{equation}

Using this and Lemmas \ref{hc:lm:ell} and \ref{hc:lm:algP} we can conclude:

\begin{equation}
\frac{ALG_P + ALG_C}{nW} \geq \left(1-\frac{R}{4\gamma}\right) R + \alpha_{GW} \cdot \left(1-\frac{R}{2\gamma}\right) \left((1-R) - (1-R^*)\frac{3\alpha_{GW}-2}{3\alpha_{GW}}\right). \label{hc:lm:ALGpc}
\end{equation}

Let us define function $\tilde{F}(R) = \left(1-\frac{R}{4\gamma}\right) R + \alpha_{GW} \cdot \left(1-\frac{R}{2\gamma}\right) \left((1-R) - (1-R^*)\frac{3\alpha_{GW}-2}{3\alpha_{GW}}\right)$. Note that the equation above says $\frac{ALG_P + ALG_C}{nW} \geq \tilde{F}(R)$ and:

\begin{eqnarray*}
	\tilde{F}(R^*) &=& \left(1-\frac{R^*}{4\gamma}\right) R^* + \alpha_{GW} \left(1-\frac{R^*}{2\gamma}\right) \left((1-R^*) - (1-R^*)\frac{3\alpha_{GW}-2}{3\alpha_{GW}}\right) \\
	&=& \left(1-\frac{R^*}{4\gamma}\right) R^* + \alpha_{GW} \left(1-\frac{R^*}{2\gamma}\right) (1-R^*) \left(1- \frac{3\alpha_{GW}-2}{3\alpha_{GW}}\right)\\
	&=& \left(1-\frac{R^*}{4\gamma}\right) R^* + \alpha_{GW} \left(1-\frac{R^*}{2\gamma}\right) (1-R^*) \frac{2}{3\alpha_{GW}}\\
	&=& \left(1-\frac{R^*}{4\gamma}\right) R^* + \frac{2}{3}\left(1-\frac{R^*}{2\gamma}\right) (1-R^*)\\
	&=& F(R^*)
\end{eqnarray*}

So, the only thing we need to show to complete the proof of Theorem \ref{hc:thm:algPC} is that
$\tilde{F}(R)$ is decreasing with $R$ 
increasing in interval $0 \leq R \leq R^*$ and the following lemma will prove this (see Appendix \ref{app:proofs}):
\begin{lemma} \label{hc:lm:c2}
	If Condition (\ref{hc:cond:Rg}) holds, then $\tilde{F}(R^*) \leq \tilde{F}(R)$ for all $0 \leq R \leq R^*$.
\end{lemma}

It is easy to see that Equation (\ref{hc:lm:ALGpc}) together with Lemma \ref{hc:lm:c2} imply Theorem \ref{hc:thm:algPC}.

\subsubsection{Using Theorems \ref{hc:thm:algPR} and \ref{hc:thm:algPC} to prove Theorem \ref{hc:thm:DissimMax}}

To complete the proof of Theorem \ref{hc:thm:DissimMax},
we use Theorems \ref{hc:thm:algPR} and \ref{hc:thm:algPC} to conclude that
if we can set $\gamma$ and $R^*$ such that Conditions (\ref{hc:cond:Re}),
(\ref{hc:cond:Reg}), and (\ref{hc:cond:Rg}) are met,
then the ratio of our algorithm would be at least $F(R^*)$:

\begin{equation}\label{eqn:FR}
	F(R^*) = \left(1-\frac{R^*}{4\gamma}\right) R^*+\frac{2}{3}\left(1-\frac{R^*}{2\gamma}\right) (1-R^*).
\end{equation}

Recall that we can assume $\epsilon<\frac{\alpha-2/3}{\alpha}$ as otherwise the random partitioning gives a
better than $\alpha$-approximation. 
If we fix a value for $\alpha$, this condition implies a bound for $\epsilon$, which then 
using Condition (\ref{hc:cond:Reg}) gives a bound for $\gamma$:
$\gamma = (1-R^*)^2 \cdot \frac{3\alpha_{GW}-2}{9\alpha_{GW}^2} \cdot \frac{\alpha}{\alpha-2/3}$
and this in turn implies the best value for $R^*$ using the equation above for $F(R^*)$.
Note that the ratio of our algorithm would be the minimum of $\alpha$ and $F(R^*)$.
It is easy to verify that by having $\alpha = 0.6929$, we will get $R^* = 0.227617$ 
which maximizes $F(.)$ exactly at $F(R^*)=\alpha$; for this value $\gamma$ will be set to $1.442042$.
Note that these values of $\alpha$, $R^*$ and $\gamma$ will satisfy all the conditions of Theorem \ref{hc:thm:algPC} and the ratio of algorithm
will be $0.6929$. This completes the proof of Theorem \ref{hc:thm:DissimMax}.

\subsection{Improving the ratio to 0.71604: Proof of Theorem \ref{hc:max-dissim}}
In this section we will show how to improve the ratio of our algorithm to $0.71604$. 
A key observation from the past section is that
while we use $\epsilon$ in our conditions, at the end we are comparing the revenue obtained by the algorithm
with $nW$ and not with $OPT = (1-\epsilon) nW$. In other words, the function $F(.)$ is obtained based on comparing
the solution of the algorithm with $nW$ (in both Theorems \ref{hc:thm:algPR} and \ref{hc:thm:algPC}).
This is because we do not have any lower-bound on $\epsilon$ and it really could be arbitrarily close to $0$.
However, when $\epsilon$ is close to $0$, then our conditions, especially Condition (\ref{hc:cond:Reg}),
will let us choose better $\gamma$ and $R^*$. We will take advantage of this to find a series of values for
$\gamma$ and $R^*$, such that if we run our algorithm with these parameters and take the better of all solutions
then the ratio will be $0.71604$.

Recall that to maximize the ratio of our algorithm we first set $\alpha$ and use our conditions to find the 
best $R^*$ and $\gamma$ to maximize $F(R^*)$ (using Equation (\ref{eqn:FR})) 
and finally the ratio would be the minimum of $\alpha$ and $F(R^*)$.
If $F(R^*) < \alpha$, it means that our initial choice of $\alpha$ was too high and vise versa.
Suppose we set $\alpha$ to some value such that based on that value we obtain 
$R^*_1$ and $\gamma_1$ that are maximizing $F(R^*_1)$ but $F(R^*_1) < \alpha$.
As we mentioned earlier, since the actual ratio of our algorithm is $\frac{F(R^*_1)}{1-\epsilon}$ (since we assume
$OPT=(1-\epsilon)nW$), there is an $\epsilon_2<\epsilon$ 
where $\frac{F(R^*_1)}{1-\epsilon_2} = \alpha$ and for all values of $\epsilon$ where
$\epsilon_2 \leq \epsilon \leq \epsilon_1$ (we define $\epsilon_1=\frac{\alpha-2/3}{\alpha}$)
the algorithm with parameters $R^*_1$ and $\gamma_1$
is actually an $\alpha$-approximation and all three conditions of Theorem \ref{hc:thm:algPC} are met.
But what happens if $\epsilon < \epsilon_2$? In this case the bounds of those three conditions are in fact
better and we can set the parameters $R^*$ and $\gamma$ differently to obtain better ratios.

As an example, suppose that we want to improve the ratio to $\alpha = 0.7$.
Using the bound $\epsilon \leq \epsilon_1=\frac{\alpha-2/3}{\alpha} = 1/21$, we use Condition
(\ref{hc:cond:Reg}) and set $\gamma = (1-R^*)^2 \cdot \frac{3\alpha_{GW}-2}{9\alpha_{GW}^2} \cdot 21 = 1.9218297 \cdot (1-R^*)^2$ and find the best $R^*$ to maximize $F(R^*)$ in Equation (\ref{eqn:FR}).
It turns out the maximum of $F(.)$ is obtained at 
$R^*_1 = 0.173114$ and $\gamma_1 = 1.314032$, with $F(R^*_1) = 0.682358$.
This means that by running our algorithm with parameters $R^*_1$ and $\gamma_1$ we will have the revenue at least $0.682358 nW$.
To make this an $\alpha$-approximation (for $\alpha = 0.7$), we must have a lower-bound on $\epsilon$.
More precisely $0.682358 nW \geq \alpha OPT = 0.7 \cdot (1-\epsilon) nW$ only if $\epsilon \geq \epsilon_2 = 0.02520285$.
This means that if $\epsilon_2 \leq \epsilon \leq \epsilon_1$, then 
the revenue gained by our algorithm using parameters $R^*_1$ and $\gamma_1$ is at least $\alpha OPT$.
Recall that if $\epsilon > \epsilon_1$, then random partitioning is an $\alpha$-approximation.
The only remaining case is
if $\epsilon < \epsilon_2$, and in this case we can set $\gamma = (1-R^*)^2 \cdot \frac{3\alpha_{GW}-2}{9\alpha_{GW}^2} \cdot \frac{1}{\epsilon_2} = 3.6311637 (1-R^*)^2$ and again find the best $R^*$ to maximize $F(R^*)$.
We must set $R^*_2 = 0.316719$ and $\gamma_2 = 1.695292$ to maximize $F(R^*_2) = 0.714896$ which is even greater than $\alpha$. This means that if $\epsilon < \epsilon_2$, 
then the revenue gained by our algorithm using parameters $R^*_2$ and $\gamma_2$ is at least $\alpha OPT$.
Note that for $\gamma_1,R^*_1$ and $\epsilon_1=\frac{\alpha-2/3}{\alpha}$ and for
$\gamma_2,R^*_2,\epsilon_2$ all three conditions of Theorem \ref{hc:thm:algPC} are met.
Because we do not know the exact value of $\epsilon$, the only thing we need to do is to take the better of the following three:
\begin{enumerate}
	\item Run Random Partitioning Algorithm all the way.
	\item Run Algorithm \ref{hc:alg:peel-off-max-cut-random} with parameters $R^*_1$ and $\gamma_1$.
	\item Run Algorithm \ref{hc:alg:peel-off-max-cut-random} with parameters $R^*_2$ and $\gamma_2$.
\end{enumerate}

So, with only two sets of parameters for $R^*$ and $\gamma$ and running Algorithm \ref{hc:alg:peel-off-max-cut-random}
for each (and taking the better of the results as well as random partitioning) we can get a 0.7-approximation.
It turns out using this approach and running the algorithm with several more parameters we can get slight improvement.
More specifically, starting with $\alpha=0.716$ we will find 83 triples of values $\gamma_i,R^*_i,\epsilon_i$ 
($1\leq i\leq 83$) such that for each triple, the three conditions of Theorem \ref{hc:thm:algPC} are met 
and for each pair of $\gamma_i,R^*_i$ if we run Algorithm \ref{hc:alg:peel-off-max-cut-random} with these parameters
if the actual value of $\epsilon$ is between $\epsilon_{i+1}$ and $\epsilon_{i}$ (with the assumption of $\epsilon_{84}=0$) then the revenue of the solution is at least $\alpha OPT$.
These 83 triples of $\gamma_i,R^*_i,\epsilon_i$ are obtained using a simple computer program and are listed in Table \ref{hc:table:716} in Appendix \ref{app:parameters}.
By choosing $\alpha=0.71604$ we will have 211 triples of $\gamma_i,R^*_i,\epsilon_i$ and ratio will be at least
$0.71604$.

\section{Proof of Theorem \ref{hc:thm:h}}

We discuss a method for bounding the cost of any top-down algorithm regarding this new objective function (\ref{dissim-hij}).
When a top-down algorithm splits cluster $A \cup B$ into clusters $A$ and $B$, 
the least common ancestor for any two pair of nodes $a \in A$ and $b \in B$ is determined.
But for each two nodes $a,a' \in A$ (also for $b,b' \in B$),
the distance between their least common ancestor and the root in the final tree is increased by 1.
Given this, we define the cost of partitioning a cluster into two clusters $A$ and $B$ as the following formula:
\begin{equation*}
	split\text{-}cost_H(A,B) = \sum_{a,a' \in A} w_{a,a'} + \sum_{b,b' \in B} w_{b,b'}.
\end{equation*}

Notice that the total cost $Cost_H(T)$ is exactly the sum of the $split$-$cost$ over all internal nodes of tree $T$ plus
and additional $\sum_{i,j \in [n]} w_{i,j}$.
We model the problem as a multi-step game. At each step we have to choose a cut $(S_1, S_2)$ in $G$ 
and remove its cutting edges $\delta(S_1) = \delta(S_2)$.
The cost of step $0$ is $W$ and the cost of each step is the total weight of all the remaining edges.
Then the total cost of the algorithm would be summation over all iterations.
Considering this view, it would make sense to choose the maximum weighted cut at each step.
One of our main results is to analyze this algorithm by showing that it has a constant approximation ratio.
Note that this algorithm has a large approximation ratio when we want to minimize 
Dasgupta's objective function $Cost(T)$ or maximize its dual $Rev_{Dual}(T)$.

The method we use to prove Theorem \ref{hc:thm:h} is inspired by the work done by Feige et al.~\cite{FLT04}.
We assume all the weights are integers and for each edge $(i,j) \in E$ we replace it with $w_{i,j}$ many unit weight edges to make the graph unweighted.
Let $OPT$ be the cost of the optimal tree and $AMC$ be the cost of the approximated Max-Cut algorithm.
For $i = 1, 2, ...$ let $X_i^*$ and $X_i$ be the sets of edges removed at step $i$ by the optimal tree and approximated Max-Cut algorithm, respectively.
Also, let $R_i^*$ and $R_i$ be the sets of edges remained after step $i$ of the algorithms. We also set $R_0^* = R_0 = E$.
Notice that for each $i \geq 0$ we have $R_i = E \setminus \cup_{j=1}^{i-1} X_i$, 
$R_i = X_{i+1} \cup R_{i+1}$ and $R_i = \cup_{j=i+1}^n X_j$ (same for $R_i^*$).

Now observe that we have:
\begin{equation*}
	OPT = \sum_{i=0}^n \vert R_i^* \vert = \sum_{i=1}^n i \cdot \vert X_i^* \vert \quad \quad \text{ and } \quad \quad AMC = \sum_{i=0}^n \vert R_i \vert = \sum_{i=1}^n i \cdot \vert X_i \vert.
\end{equation*}

Now we define $p_1 = 0$ and for each $i \geq 2$ we define $p_i = \frac{\vert R_{i-1} \vert}{\vert X_i \vert}$ and for each $e \in E$ we set $p_e = p_i$ if $e \in X_i$, we then have:
\begin{equation*}
	\sum_{e \in E} p_e = \sum_{i=1}^n \sum_{e \in X_i} p_i = \sum_{i = 1}^n \vert X_i \vert \cdot p_i = \sum_{i = 2}^n \vert X_i \vert \cdot \frac{\vert R_{i-1} \vert}{\vert X_i \vert} = \sum_{i=1} \vert R_i \vert = AMC - \vert R_0 \vert.
\end{equation*}

Observe that if you use an algorithm that chooses a cut that contains at least half of the edges at each step, 
including approximated maximum cut, then for each $e \in E \setminus X_1$ we have $p_e \leq 2$.
Note that $p_e = p_1 = 0$ for each $e \in X_1$, so we have the following upper-bound for $AMC$:
\begin{equation*}
	AMC - \lvert R_0 \rvert = \sum_{e \in E \setminus X_1} p_e \leq \sum_{e \in E \setminus X_1}^n 2 = 2 \lvert R_1 \rvert \quad \quad \Rightarrow \quad \quad AMC \leq \lvert R_0 \rvert + 2 \lvert R_1 \rvert.
\end{equation*}

Now we consider two cases. For the first case we assume $\lvert R_1^* \rvert > \frac{2\alpha_{GW} - 1}{2\alpha_{GW}} \cdot \lvert R_0 \rvert$,
and the second case is when $\lvert R_1^* \rvert \leq \frac{2\alpha_{GW} - 1}{2\alpha_{GW}} \cdot \lvert R_0 \rvert$, where $\alpha_{GW}$ is the ratio of
approximate max-cut algorithm.

For the first case, note that $X_1$ is the cut the approximated maximum cut chooses in the first step
and it has at least half of the edges, so $\lvert R_1 \rvert \leq \frac{\lvert E \rvert}{2} = \frac{\lvert R_0 \rvert}{2}$.
If use the lower bound of $\lvert R_0 \rvert + \lvert R_1^* \rvert$ for $OPT$, then we have:
\begin{equation*}
	\frac{AMC}{OPT} \leq \frac{\lvert R_0 \rvert + 2 \lvert R_1 \rvert}{\lvert R_0 \rvert + \lvert R_1^* \rvert} < \frac{\lvert R_0 \rvert + \lvert R_0 \rvert}{\lvert R_0 \rvert + \frac{2\alpha_{GW} - 1}{2\alpha_{GW}} \lvert R_0 \rvert} = \frac{2}{1+\frac{2\alpha_{GW} - 1}{2\alpha_{GW}}} = \frac{4\alpha_{GW}}{4\alpha_{GW}-1}.
\end{equation*}

For the second case, note that $X_1^*$ is a cut in $G$ and $X_1$ has at least $\alpha_{GW}$ fraction of the maximum cut or any other cut including $X_1^*$, so:
\begin{equation*}
	\lvert X_1 \rvert \geq \alpha_{GW} \cdot \lvert X_1^* \rvert.
\end{equation*}

So, for $R_1$ we have:
\begin{equation*}
	\lvert R_1 \rvert = \vert R_0 \vert - \vert X_1 \vert \leq \vert R_0 \vert - \alpha_{GW} \cdot \lvert X_1^* \rvert = \vert R_0 \vert - \alpha_{GW} \cdot (\vert R_0 \vert - \vert R_1^* \vert) = (1-\alpha_{GW}) \vert R_0 \vert + \alpha_{GW} \cdot \vert R_1^* \vert.
\end{equation*}

Then again, we use the lower bound of $\lvert R_0 \rvert + \lvert R_1^* \rvert$ for $OPT$ to bound the ratio of the Approximated Max-Cut algorithm:
\begin{equation*}
	\frac{AMC}{OPT} \leq \frac{\lvert R_0 \rvert + 2 \lvert R_1 \rvert}{\lvert R_0 \rvert + \lvert R_1^* \rvert} \leq \frac{\lvert R_0 \rvert + 2(1-\alpha_{GW}) \vert R_0 \vert + 2\alpha_{GW} \cdot \vert R_1^* \vert}{\lvert R_0 \rvert + \lvert R_1^* \rvert} = \frac{(3-2\alpha_{GW}) \vert R_0 \vert + 2\alpha_{GW} \cdot \vert R_1^* \vert}{\lvert R_0 \rvert + \lvert R_1^* \rvert}.
\end{equation*}

Remember that in this case $\lvert R_0 \rvert \geq \frac{2\alpha_{GW}}{2\alpha_{GW} - 1} \cdot \lvert R_1^* \rvert$, so:
\begin{equation*}
	\frac{AMC}{OPT} \leq (3-2\alpha_{GW}) + \frac{(4\alpha_{GW}-3) \cdot \vert R_1^* \vert}{\lvert R_0 \rvert + \lvert R_1^* \rvert} \leq (3-2\alpha_{GW}) + \frac{(4\alpha_{GW}-3) \cdot \vert R_1^* \vert}{\frac{2\alpha_{GW}}{2\alpha_{GW} - 1} \cdot \lvert R_1^* \rvert + \lvert R_1^* \rvert}.
\end{equation*}
\begin{equation*}
	\Rightarrow \quad \frac{AMC}{OPT} \leq (3-2\alpha_{GW}) + \frac{(4\alpha_{GW}-3)(2\alpha_{GW}-1)}{(4\alpha_{GW}-1)} = \frac{4\alpha_{GW}}{4\alpha_{GW}-1}.
\end{equation*}
	
Thus, in either case, the cost of the solution returned by the recursively finding an $\alpha_{GW}$-approximate max-cut is at most $\frac{4\alpha_{GW}}{4\alpha_{GW}-1}$
times the optimum. This completes the proof of Theorem \ref{hc:thm:h}.

\bibliographystyle{abbrv}
\bibliography{HC-v9}

\appendix
\section{Missing Proofs}\label{app:proofs}
\begin{pproof}{Lemma \ref{hc:lm:c}}
To prove this lemma, we need to show that $G(c)$ is increasing with $c$. And note that a function like $\frac{X-Yc}{1-Zc}$,
where $X, Y$ and $Z$ are all independent of $c$, is ascending if and only if $XY \geq Z$.
In $G(c)$ we have $X = (1-R) - (1-R^*)\frac{3\alpha_{GW}-2}{3\alpha_{GW}}$, $Y = 2\gamma$ and $Z = \frac{(1-R^*)}{\epsilon} \frac{3\alpha_{GW}-2}{3\alpha_{GW}}$.
Now recall that we are in a case where $R \leq R^*$, so:

\begin{equation*}
	X \geq (1-R^*) - (1-R^*)\frac{3\alpha_{GW}-2}{3\alpha_{GW}} = (1-R^*) \cdot (1 - \frac{3\alpha_{GW}-2}{3\alpha_{GW}}) = (1-R^*) \frac{2}{3\alpha_{GW}}.
\end{equation*}

Now using the definition of $X,Y,Z$ and the bound above for $X$, Condition (\ref{hc:cond:Reg}) implies $XZ \geq Y$, which completes the proof of the lemma.
\end{pproof}

\begin{pproof}{Lemma \ref{hc:lm:c2}}
First we do some simplifications to $\tilde{F}(R)$:

\begin{eqnarray*}
	\tilde{F}(R) &=& \left(1-\frac{R}{4\gamma}\right) R + \alpha_{GW} \left(1-\frac{R}{2\gamma}\right) \left((1-R) - (1-R^*)\frac{3\alpha_{GW}-2}{3\alpha_{GW}}\right)\\
	&=& \alpha_{GW} - (1-R^*)(\alpha_{GW}-2/3) - R \cdot \left( \alpha_{GW} + \frac{\alpha_{GW}}{2\gamma} - 1 - \frac{(1-R^*)(\alpha_{GW}-2/3)}{2\gamma} \right) + R^2 \cdot \left(\frac{\alpha_{GW}}{2\gamma} - \frac{1}{4\gamma}\right).
\end{eqnarray*}

Let $X=\alpha_{GW} - (1-R^*)(\alpha_{GW}-2/3)$, $Y=\left( \alpha_{GW} + \frac{\alpha_{GW}}{2\gamma} - 1 - \frac{(1-R^*)(\alpha_{GW}-2/3)}{2\gamma} \right)$, and $Z=\left(\frac{\alpha_{GW}}{2\gamma} - \frac{1}{4\gamma}\right)$;
then $\tilde{F}(R)=X-YR+ZR^2$. Observe that since $\gamma>0$, $Z$ is positive.
Also, we will soon show that $Y$ is positive too. Thus, the minimum of $\tilde{F}(R)$ is at $R=\frac{Y}{2Z}$.
If we show that $R^*\leq \frac{Y}{2Z}$ then we have shown that $\tilde{F}(R)$ is
totally descending in the interval $0 \leq R \leq R^*$. So, it is enough to show:

\begin{equation*}
	R^* \leq \frac{\alpha_{GW} + \frac{\alpha_{GW}}{2\gamma} - 1 - \frac{(1-R^*)(\alpha_{GW}-2/3)}{2\gamma}}{2 (\frac{\alpha_{GW}}{2\gamma} - \frac{1}{4\gamma})}.
\end{equation*}

Or equivalently we need to have:
\begin{equation*}
 R^*(\alpha_{GW}-1/2) \leq R^*(\alpha_{GW}/2-1/3) + 1/3 - \gamma (1-\alpha_{GW}).
\end{equation*}

This inequality is exactly what we have in Condition (\ref{hc:cond:Rg}).
Also, since this condition implies that $R^*< \frac{Y}{2Z}$ and we know that $R^*>0$, this also means $Y>0$.
Thus, the minimum of $\tilde{F}(R)$ is at $R=R^*$ and this completes the proof of the lemma.
\end{pproof}
\section{$R^*$ and $\gamma$ Values to have a $0.716$-approximation}\label{app:parameters}
\renewcommand{\arraystretch}{1.3}
\begin{longtable}{|c|c|c|c|}
	\caption{Values of parameters $R^*$, $\gamma$ and $\epsilon$ to run Algorithm \ref{hc:alg:peel-off-max-cut-random} and take the best and improve the approximation ratio to $0.716$}
	\label{hc:table:716}\\
	\hline
	\textbf{$R^*$} & \textbf{$\gamma$} & \textbf{$F(R^*)$} & \textbf{$\epsilon$}\\
	\hline
	$R^*_{1} = 0.07852$ & $\gamma_{1} = 1.1278206$ & $F(R^*_{1}) = 0.670089$ & $\epsilon_{1} = 0.0689014$\\
	$R^*_{2} = 0.09761$ & $\gamma_{2} = 1.1621875$ & $F(R^*_{2}) = 0.67189$ & $\epsilon_{2} = 0.0641222$\\
	$R^*_{3} = 0.10809$ & $\gamma_{3} = 1.1817295$ & $F(R^*_{3}) = 0.673031$ & $\epsilon_{3} = 0.0616056$\\
	$R^*_{4} = 0.11489$ & $\gamma_{4} = 1.1946796$ & $F(R^*_{4}) = 0.673828$ & $\epsilon_{4} = 0.0600121$\\
	$R^*_{5} = 0.11973$ & $\gamma_{5} = 1.2039729$ & $F(R^*_{5}) = 0.67442$ & $\epsilon_{5} = 0.0588994$\\
	$R^*_{6} = 0.12337$ & $\gamma_{6} = 1.2110441$ & $F(R^*_{6}) = 0.67488$ & $\epsilon_{6} = 0.0580723$\\
	$R^*_{7} = 0.12622$ & $\gamma_{7} = 1.216645$ & $F(R^*_{7}) = 0.67525$ & $\epsilon_{7} = 0.0574297$\\
	$R^*_{8} = 0.12853$ & $\gamma_{8} = 1.2211909$ & $F(R^*_{8}) = 0.675554$ & $\epsilon_{8} = 0.0569138$\\
	$R^*_{9} = 0.13045$ & $\gamma_{9} = 1.2249629$ & $F(R^*_{9}) = 0.67581$ & $\epsilon_{9} = 0.0564888$\\
	$R^*_{10} = 0.13206$ & $\gamma_{10} = 1.2282025$ & $F(R^*_{10}) = 0.676029$ & $\epsilon_{10} = 0.0561313$\\
	$R^*_{11} = 0.13346$ & $\gamma_{11} = 1.2309501$ & $F(R^*_{11}) = 0.676219$ & $\epsilon_{11} = 0.0558255$\\
	$R^*_{12} = 0.13467$ & $\gamma_{12} = 1.2333801$ & $F(R^*_{12}) = 0.676386$ & $\epsilon_{12} = 0.05556$\\
	$R^*_{13} = 0.13574$ & $\gamma_{13} = 1.2355205$ & $F(R^*_{13}) = 0.676535$ & $\epsilon_{13} = 0.0553267$\\
	$R^*_{14} = 0.13669$ & $\gamma_{14} = 1.2374426$ & $F(R^*_{14}) = 0.676668$ & $\epsilon_{14} = 0.0551194$\\
	$R^*_{15} = 0.13755$ & $\gamma_{15} = 1.2391589$ & $F(R^*_{15}) = 0.676788$ & $\epsilon_{15} = 0.0549334$\\
	$R^*_{16} = 0.13832$ & $\gamma_{16} = 1.2407467$ & $F(R^*_{16}) = 0.676898$ & $\epsilon_{16} = 0.0547652$\\
	$R^*_{17} = 0.13903$ & $\gamma_{17} = 1.2421809$ & $F(R^*_{17}) = 0.676999$ & $\epsilon_{17} = 0.0546119$\\
	$R^*_{18} = 0.13968$ & $\gamma_{18} = 1.2435107$ & $F(R^*_{18}) = 0.677092$ & $\epsilon_{18} = 0.0544711$\\
	$R^*_{19} = 0.14029$ & $\gamma_{19} = 1.2447186$ & $F(R^*_{19}) = 0.677178$ & $\epsilon_{19} = 0.0543411$\\
	$R^*_{20} = 0.14085$ & $\gamma_{20} = 1.2458665$ & $F(R^*_{20}) = 0.677259$ & $\epsilon_{20} = 0.0542204$\\
	$R^*_{21} = 0.14138$ & $\gamma_{21} = 1.2469242$ & $F(R^*_{21}) = 0.677335$ & $\epsilon_{21} = 0.0541076$\\
	$R^*_{22} = 0.14187$ & $\gamma_{22} = 1.2479438$ & $F(R^*_{22}) = 0.677406$ & $\epsilon_{22} = 0.0540017$\\
	$R^*_{23} = 0.14234$ & $\gamma_{23} = 1.2488869$ & $F(R^*_{23}) = 0.677474$ & $\epsilon_{23} = 0.0539018$\\
	$R^*_{24} = 0.14278$ & $\gamma_{24} = 1.2497992$ & $F(R^*_{24}) = 0.677539$ & $\epsilon_{24} = 0.0538072$\\
	$R^*_{25} = 0.1432$ & $\gamma_{25} = 1.2506663$ & $F(R^*_{25}) = 0.6776$ & $\epsilon_{25} = 0.0537172$\\
	$R^*_{26} = 0.14361$ & $\gamma_{26} = 1.2514715$ & $F(R^*_{26}) = 0.677659$ & $\epsilon_{26} = 0.0536313$\\
	$R^*_{27} = 0.14399$ & $\gamma_{27} = 1.2522842$ & $F(R^*_{27}) = 0.677716$ & $\epsilon_{27} = 0.0535489$\\
	$R^*_{28} = 0.14437$ & $\gamma_{28} = 1.2530263$ & $F(R^*_{28}) = 0.67777$ & $\epsilon_{28} = 0.0534697$\\
	$R^*_{29} = 0.14473$ & $\gamma_{29} = 1.2537652$ & $F(R^*_{29}) = 0.677823$ & $\epsilon_{29} = 0.0533932$\\
	$R^*_{30} = 0.14508$ & $\gamma_{30} = 1.2544791$ & $F(R^*_{30}) = 0.677875$ & $\epsilon_{30} = 0.0533192$\\
	$R^*_{31} = 0.14542$ & $\gamma_{31} = 1.255175$ & $F(R^*_{31}) = 0.677925$ & $\epsilon_{31} = 0.0532473$\\
	$R^*_{32} = 0.14575$ & $\gamma_{32} = 1.2558593$ & $F(R^*_{32}) = 0.677974$ & $\epsilon_{32} = 0.0531771$\\
	$R^*_{33} = 0.14607$ & $\gamma_{33} = 1.2565379$ & $F(R^*_{33}) = 0.678022$ & $\epsilon_{33} = 0.0531086$\\
	$R^*_{34} = 0.14639$ & $\gamma_{34} = 1.2571865$ & $F(R^*_{34}) = 0.67807$ & $\epsilon_{34} = 0.0530414$\\
	$R^*_{35} = 0.1467$ & $\gamma_{35} = 1.2578398$ & $F(R^*_{35}) = 0.678116$ & $\epsilon_{35} = 0.0529754$\\
	$R^*_{36} = 0.14701$ & $\gamma_{36} = 1.2584728$ & $F(R^*_{36}) = 0.678162$ & $\epsilon_{36} = 0.0529103$\\
	$R^*_{37} = 0.14732$ & $\gamma_{37} = 1.2590901$ & $F(R^*_{37}) = 0.678208$ & $\epsilon_{37} = 0.0528459$\\
	$R^*_{38} = 0.14762$ & $\gamma_{38} = 1.2597255$ & $F(R^*_{38}) = 0.678253$ & $\epsilon_{38} = 0.0527821$\\
	$R^*_{39} = 0.14792$ & $\gamma_{39} = 1.2603536$ & $F(R^*_{39}) = 0.678299$ & $\epsilon_{39} = 0.0527187$\\
	$R^*_{40} = 0.14822$ & $\gamma_{40} = 1.2609785$ & $F(R^*_{40}) = 0.678344$ & $\epsilon_{40} = 0.0526554$\\
	$R^*_{41} = 0.14852$ & $\gamma_{41} = 1.2616042$ & $F(R^*_{41}) = 0.678389$ & $\epsilon_{41} = 0.0525923$\\
	$R^*_{42} = 0.14883$ & $\gamma_{42} = 1.2622052$ & $F(R^*_{42}) = 0.678435$ & $\epsilon_{42} = 0.0525289$\\
	$R^*_{43} = 0.14913$ & $\gamma_{43} = 1.2628447$ & $F(R^*_{43}) = 0.678481$ & $\epsilon_{43} = 0.0524653$\\
	$R^*_{44} = 0.14943$ & $\gamma_{44} = 1.2634973$ & $F(R^*_{44}) = 0.678527$ & $\epsilon_{44} = 0.0524013$\\
	$R^*_{45} = 0.14974$ & $\gamma_{45} = 1.2641378$ & $F(R^*_{45}) = 0.678574$ & $\epsilon_{45} = 0.0523366$\\
	$R^*_{46} = 0.15006$ & $\gamma_{46} = 1.2647704$ & $F(R^*_{46}) = 0.678622$ & $\epsilon_{46} = 0.052271$\\
	$R^*_{47} = 0.15038$ & $\gamma_{47} = 1.2654297$ & $F(R^*_{47}) = 0.67867$ & $\epsilon_{47} = 0.0522044$\\
	$R^*_{48} = 0.1507$ & $\gamma_{48} = 1.2661208$ & $F(R^*_{48}) = 0.67872$ & $\epsilon_{48} = 0.0521367$\\
	$R^*_{49} = 0.15103$ & $\gamma_{49} = 1.2668193$ & $F(R^*_{49}) = 0.678771$ & $\epsilon_{49} = 0.0520674$\\
	$R^*_{50} = 0.15138$ & $\gamma_{50} = 1.2675011$ & $F(R^*_{50}) = 0.678823$ & $\epsilon_{50} = 0.0519965$\\
	$R^*_{51} = 0.15173$ & $\gamma_{51} = 1.2682323$ & $F(R^*_{51}) = 0.678876$ & $\epsilon_{51} = 0.0519237$\\
	$R^*_{52} = 0.15209$ & $\gamma_{52} = 1.26899$ & $F(R^*_{52}) = 0.678932$ & $\epsilon_{52} = 0.0518486$\\
	$R^*_{53} = 0.15246$ & $\gamma_{53} = 1.2697821$ & $F(R^*_{53}) = 0.67899$ & $\epsilon_{53} = 0.0517711$\\
	$R^*_{54} = 0.15285$ & $\gamma_{54} = 1.2705873$ & $F(R^*_{54}) = 0.679049$ & $\epsilon_{54} = 0.0516907$\\
	$R^*_{55} = 0.15325$ & $\gamma_{55} = 1.2714451$ & $F(R^*_{55}) = 0.679112$ & $\epsilon_{55} = 0.051607$\\
	$R^*_{56} = 0.15368$ & $\gamma_{56} = 1.2723064$ & $F(R^*_{56}) = 0.679177$ & $\epsilon_{56} = 0.0515197$\\
	$R^*_{57} = 0.15412$ & $\gamma_{57} = 1.273244$ & $F(R^*_{57}) = 0.679246$ & $\epsilon_{57} = 0.0514283$\\
	$R^*_{58} = 0.15459$ & $\gamma_{58} = 1.2742119$ & $F(R^*_{58}) = 0.679319$ & $\epsilon_{58} = 0.0513321$\\
	$R^*_{59} = 0.15508$ & $\gamma_{59} = 1.2752568$ & $F(R^*_{59}) = 0.679396$ & $\epsilon_{59} = 0.0512306$\\
	$R^*_{60} = 0.1556$ & $\gamma_{60} = 1.2763678$ & $F(R^*_{60}) = 0.679478$ & $\epsilon_{60} = 0.051123$\\
	$R^*_{61} = 0.15616$ & $\gamma_{61} = 1.2775371$ & $F(R^*_{61}) = 0.679566$ & $\epsilon_{61} = 0.0510085$\\
	$R^*_{62} = 0.15676$ & $\gamma_{62} = 1.2787913$ & $F(R^*_{62}) = 0.67966$ & $\epsilon_{62} = 0.0508861$\\
	$R^*_{63} = 0.1574$ & $\gamma_{63} = 1.280162$ & $F(R^*_{63}) = 0.679762$ & $\epsilon_{63} = 0.0507544$\\
	$R^*_{64} = 0.1581$ & $\gamma_{64} = 1.2816267$ & $F(R^*_{64}) = 0.679872$ & $\epsilon_{64} = 0.0506122$\\
	$R^*_{65} = 0.15886$ & $\gamma_{65} = 1.2832312$ & $F(R^*_{65}) = 0.679993$ & $\epsilon_{65} = 0.0504577$\\
	$R^*_{66} = 0.15969$ & $\gamma_{66} = 1.2850016$ & $F(R^*_{66}) = 0.680126$ & $\epsilon_{66} = 0.0502888$\\
	$R^*_{67} = 0.16061$ & $\gamma_{67} = 1.2869469$ & $F(R^*_{67}) = 0.680274$ & $\epsilon_{67} = 0.0501029$\\
	$R^*_{68} = 0.16163$ & $\gamma_{68} = 1.2891241$ & $F(R^*_{68}) = 0.680439$ & $\epsilon_{68} = 0.0498968$\\
	$R^*_{69} = 0.16277$ & $\gamma_{69} = 1.2915835$ & $F(R^*_{69}) = 0.680625$ & $\epsilon_{69} = 0.0496664$\\
	$R^*_{70} = 0.16407$ & $\gamma_{70} = 1.2943462$ & $F(R^*_{70}) = 0.680837$ & $\epsilon_{70} = 0.0494067$\\
	$R^*_{71} = 0.16555$ & $\gamma_{71} = 1.2975399$ & $F(R^*_{71}) = 0.681081$ & $\epsilon_{71} = 0.0491107$\\
	$R^*_{72} = 0.16727$ & $\gamma_{72} = 1.3012328$ & $F(R^*_{72}) = 0.681366$ & $\epsilon_{72} = 0.0487696$\\
	$R^*_{73} = 0.16928$ & $\gamma_{73} = 1.305618$ & $F(R^*_{73}) = 0.681704$ & $\epsilon_{73} = 0.0483715$\\
	$R^*_{74} = 0.17169$ & $\gamma_{74} = 1.310843$ & $F(R^*_{74}) = 0.682112$ & $\epsilon_{74} = 0.0478995$\\
	$R^*_{75} = 0.17461$ & $\gamma_{75} = 1.3172771$ & $F(R^*_{75}) = 0.682614$ & $\epsilon_{75} = 0.0473301$\\
	$R^*_{76} = 0.17825$ & $\gamma_{76} = 1.3253364$ & $F(R^*_{76}) = 0.68325$ & $\epsilon_{76} = 0.0466283$\\
	$R^*_{77} = 0.18292$ & $\gamma_{77} = 1.3357437$ & $F(R^*_{77}) = 0.68408$ & $\epsilon_{77} = 0.0457406$\\
	$R^*_{78} = 0.18912$ & $\gamma_{78} = 1.3497665$ & $F(R^*_{78}) = 0.68521$ & $\epsilon_{78} = 0.0445811$\\
	$R^*_{79} = 0.19775$ & $\gamma_{79} = 1.3696989$ & $F(R^*_{79}) = 0.686838$ & $\epsilon_{79} = 0.0430022$\\
	$R^*_{80} = 0.21061$ & $\gamma_{80} = 1.4001298$ & $F(R^*_{80}) = 0.689369$ & $\epsilon_{80} = 0.0407297$\\
	$R^*_{81} = 0.23172$ & $\gamma_{81} = 1.4523369$ & $F(R^*_{81}) = 0.693804$ & $\epsilon_{81} = 0.0371936$\\
	$R^*_{82} = 0.27259$ & $\gamma_{82} = 1.5620635$ & $F(R^*_{82}) = 0.703325$ & $\epsilon_{82} = 0.0309996$\\
	$R^*_{83} = 0.32069$ & $\gamma_{83} = 1.7081291$ & $F(R^*_{83}) = 0.716$ & $\epsilon_{83} = 0.0177022$\\
	\hline
\end{longtable}

\end{document}